\documentclass{article}
\usepackage{geometry}
\usepackage{amsfonts, amsmath, amssymb, amsthm}

\usepackage{graphicx} 

\usepackage{listings}
\usepackage{xcolor}

\usepackage{hyperref}
\usepackage{cleveref}
\usepackage{mathrsfs}

\usepackage{algorithm}
\usepackage{algpseudocode}

\newtheorem{theorem}{Theorem}[section]
\newtheorem{lemma}[theorem]{Lemma}

\crefname{conjecture}{Conjecture}{Conjectures}

\theoremstyle{definition}
\newtheorem{definition}{Definition}[section]

\DeclareMathOperator*{\E}{\mathbb{E}}

\newcommand{\Insert}{\textsf{Insert}}

\newcommand{\Occupy}{\mathsf{Occupy}}

\title{A Sharp Analysis of Fixed Offset Hash Tables\thanks{Supported by NSF Grant CCF-2446604.}}

\author{Yuhao Guo\\
{\small IIIS, Tsinghua University}
\and
Seth Pettie\\
{\small University of Michigan}
\and
Chengzhang Wan\\
{\small IIIS, Tsinghua University}
}

\date{}

\begin{document}
\maketitle
\begin{abstract}
    We consider the class of open addressing hash tables that uses a \emph{fixed} offset sequence $r_0,r_1,\ldots,r_{n-1}$ to find an open slot in the table.  As special cases this subsumes linear probing ($r_i=i$)
    and quadratic probing ($r_i=i^2$). 

    \medskip 

    We give a simple analysis that the expected insertion time when the load factor is 
    $1-\epsilon$ is $O(\epsilon^{-2})$.  This bound is tight for linear probing (but not for quadratic probing), 
    so the analysis cannot be improved at 
    this level of generality.
\end{abstract}

\section{Introduction}

We consider a class of open addressing hash tables defined by a \emph{fixed offset} probe sequence, which includes linear probing~\cite{Knuth63} and quadratic probing~\cite{Maurer68}
as special cases.  
A hash table in this class is defined by a 
permutation $(r_0,r_1,\ldots,r_{n-1})$
of $[n]$, indexing the slots of the table.  
An $\Insert(x)$ operation places 
$x\in U$ into the first slot among

\[
h(x), \; h(x) + r_1, \;\ldots, \;h(x) + r_{n-1},
\; \text{(modulo $n$)}
\]
where $h : U\to [n]$ is the hash function
and $r_0=0$ without loss of generality.

Whereas the behavior of linear 
probing has been understood for 
decades~\cite{Knuth63,Peterson57,SchayS62,KonheimW66,MendelsonY80,FlajoletPV98,Janson01}, 
there were essentially no theoretical results on 
quadratic probing and general fixed offset probing until 
very recently.
Kuszmaul and Xi~\cite{KuszmaulX24} proved that all fixed offset hash tables have constant expected insertion time up to load factor 
$0.089$ (8.9\%).  
Guo, Pettie, and Wan~\cite{GuoPW27} 
improved this load factor limit
to $35.74\%$ for any fixed offset probe 
sequence $(r_i)$ and to $37.61\%$ 
for the specific 
case of quadratic
probing $(r_i=i^2)$.
Zamir~\cite{Zamir26} proved that at load factor $1-\epsilon$,
any fixed offset hash table has $O(\epsilon^{-2}\log n)$ expected insertion cost.
Hu, Kuszmaul, Liang, Walzer, Yu, and Zhou~\cite{HuKLWYZ26b} 
recently proved that for quadratic probing
(but not other fixed offset schemes)
the expected insertion
at load factor 
$1-\epsilon$, $\epsilon > 1/\sqrt{\log n}$, 
is $\epsilon^{-(1+o(1))}$,
nearly reaching the conjectured 
bound of $O(\epsilon^{-1})$.
The $\epsilon^{-(1+o(1))}$ 
probe bound of 
Hu et al.~\cite{HuKLWYZ26b} 
is superior to the 
$\Theta(\epsilon^{-2})$ bound of 
linear probing~\cite{Knuth63}, 
though they are comparable in terms of 
cache misses, as they typically 
touch a \emph{local} area of the array of 
width $\approx \epsilon^{-2}$.
Zamir~\cite{Zamir26} proved that no open addressing
hash table can have locality $o(\epsilon^{-2})$.

\subsection{New Result}
In this short note we prove that the expected insertion 
cost for any fixed offset hash table at load factor $1-\epsilon$
is at most $1 + 2\epsilon^{-2}$.
This is the strongest statement that can be made, 
asymptotically, since the expected 
insertion cost of linear probing is $\Theta(\epsilon^{-2})$.
Even when one considers constant factors,
it cannot be improved by more than a factor of 4,
as the cost of linear probing is precisely 
$\frac{1}{2}(1 +\epsilon^{-2})$~\cite{Knuth63}.

Our analysis is unlike the recent work on quadratic and fixed offset probing~\cite{KuszmaulX24,GuoPW27,HuKLWYZ26,Zamir26}, 
and unlike prior analyses of linear probing~\cite{Knuth63,Peterson57,SchayS62,KonheimW66,MendelsonY80,FlajoletPV98,Janson01,PaghPR09,PatrascuT16}.
It is elementary, though counterintuitive in some ways.

\subsection{Organization}
\Cref{sect:prelim} introduces 
some notation and conventions
concerning sequences and sets 
of sequences.
\Cref{sect:occupation} defines
the $\Occupy_S(A)$ function, 
which counts how many slots of $S$
are occupied by insertion sequence $A$.
\Cref{sect:recursive-bound} proves
a recurrence for the number
of sequences that fully occupy $S$,
and \Cref{sect:expected-insertion-time} applies
this recurrence to upper bound the expected insertion time.

\section{Notation, Definitions, and Conventions}\label{sect:prelim}

The set of occupied slots in the hash table after inserting $(x_1,\ldots,x_m) \in U^m$ is uniquely determined by
$(h(x_1),\ldots,h(x_m)) \in [n]^m$
and the fixed offset permutation
$(r_0,r_1,\ldots,r_{n-1})$,
so we shall regard the insertion sequence itself as an 
element of $[n]^m$.
We assume $h : U\to [n]$ 
is uniformly random,
so all such sequences occur with probability $n^{-m}$.
If $x_j$ is ultimately inserted into slot $h(x_j) + r_i$, the cost of the insertion is $i+1$.

Lower case 
letters $a,b,c\in [n]$ are usually elements, 
upper case letters 
are either 
sequences $A,B,C \in [n]^*$
or sets of elements, 
and calligraphic letters
$\mathscr{A},\mathscr{B},\mathscr{C} \subset [n]^*$
are sets of sequences.
Catenation between a sequence and an element
or catenation of two sequences is indicated by juxtaposition,
e.g., $Ab$ or $AB$.
If $A=(a_1,a_2,\ldots,a_{|A|})$ is a sequence,
$A^{<i}$ and $A^{>i}$ are the sequences
$(a_1,\ldots,a_{i-1})$ and $(a_{i+1},\ldots,a_{|A|})$, respectively, while 
$A^{-i} = A^{<i}A^{>i} = (a_1,\ldots,a_{i-1},a_{i+1},\ldots,a_{|A|})$ is $A$ after removing 
the $i$th element.

\section{Properties of Occupation}\label{sect:occupation}

Much of the analysis revolves around the \emph{occupation} function $\Occupy$, 
which implicitly depends on the fixed offset permutation $(r_0,\ldots,r_{n-1})$
although the specific permutation itself 
never enters into the analysis.

\begin{definition}
    Let $S \subseteq [n]$ be a set of slots
    and $A\in [n]^*$ be an input sequence.
    $\Occupy_S(A)$ denotes the number of slots
    in $S$ that are occupied after inserting
    $A$ into the empty table.
\end{definition}

\Cref{lem:Aoccupylem} lists some basic properties 
of the $\Occupy$ function.

\begin{lemma}\label{lem:Aoccupylem}
Let $S \subseteq [n]$ be any set of slots.
\begin{enumerate}
\item For any sequence $A\in [n]^*$ 
and index $i\leq |A|$, \[
\Occupy_S(A^{-i}) 
\leq 
\Occupy_S(A) 
\leq 
\Occupy_S(A^{-i})+1.
\]
\item Let $A,B\in [n]^*$ be two sequences 
for which $A^{-i} = B^{-i}$, for some index $i$.  Then \[
|\Occupy_S(A)-\Occupy_S(B)|\leq 1.
\]
\end{enumerate}
\end{lemma}

\begin{proof}
\noindent\underline{Claim 1.}
Let $T$ and $T^{-i}$ denote the state of the hash table for the two input sequences, 
at various points in time.
Initially, both $T$ and $T^{-i}$ are empty,
and after the prefix $A^{<i}$ $T,T^{-i}$ 
are identical.
At any subsequent time, the number of occupied slots in $T$ is one more than that of $T^{-i}$.

We prove by induction that after
any prefix of $(a_{i+1},\ldots,a_{|A|})$ is inserted, the set of occupied slots in $T^{-i}$ is a subset of the occupied slots in $T$.
Suppose $s\in [n]$ is the slot occupied in 
$T$ but not $T^{-i}$ just before $a_j$, $j>i$,
is inserted.  The claim holds
if $a_j$ ultimately occupies slot $s$ in $T^{-i}$
(and hence $a_j$ occupies any other slot in $T$).
If $a_j$ ultimately occupies a slot 
$s'\neq s$ in $T^{-i}$ then the probe sequence
avoids $s$, hence $a_j$ occupies $s'$ in $T$ as well.
Either way the claim holds after inserting $a_j$.

\medskip 
\noindent\underline{Claim 2.}
By Claim 1,
\[
\Occupy_S(B) 
\leq \Occupy_S(A^{-i}) + 1
\leq \Occupy_S(A) + 1,
\]
which implies the claim.
\end{proof}

The analysis focuses on the 
relationship between the following functions.
\begin{definition}\label{def:H-F}
Define the following:
\begin{align*}
    \mathscr{H}_S(t) 
    &= \{A\in [n]^t \mid \Occupy_S(A) = |S|\},\\
    \mathscr{H}_S^{-i}(t)
    &= \{A\in [n]^t \mid \Occupy_S(A^{-i}) = |S|\},\\
    F_S(t) &= |\mathscr{H}_S(t)|,\\
    p_{k}(t) &= n^{-t}\max\{F_S(t) \mid S\subset [n] \mbox{ and } |S|=k\}.
\end{align*}
\end{definition}

In other words, $\mathscr{H}_S(t)$ is the set of length-$t$ sequences that fully occupy $S$
and $\mathscr{H}_S^{-i}(t) \subseteq \mathscr{H}_S(t)$
are those that still occupy $S$ when the $i$th element
is removed. 
$F_S(t)$ is the number of sequences in $\mathscr{H}_S(t)$,
and $p_k(t)$ is an upper bound on the probability that any particular $S$ of size $k$ is fully occupied, 
when picking a length-$t$ input sequence at random.

\begin{lemma}\label{lem:H-i}
    $|\mathscr{H}_S^{-i}(t)| = nF_S(t-1)$.
\end{lemma}

\begin{proof}
    This follows from \Cref{lem:Aoccupylem}
    and the fact that there are $n$ choices
    for the $i$th element of the sequence.
\end{proof}

\begin{lemma}\label{lem:insertioncost}
For any fixed offset permutation, 
the expected 
cost of inserting the $(t+1)$th 
element is at most 
\[
1+\sum_{k=1}^{t} p_k(t).
\]
\end{lemma}
\begin{proof}
Let $r_0,r_1,\dots,r_{n-1}$ be the fixed offset permutation, 
and $A\in [n]^m$ the uniformly
random input sequence $(a_1,\ldots,a_m)$.
Define $S_k = \{a_{t+1}+r_0,a_{t+1}+r_1,\dots,a_{t+1}+r_{k-1}\}$ 
to be the first $k$ slots in the probe sequence
for $a_{t+1}$.
The expected insertion cost is
\begin{align*}
    \E(\text{Cost of inserting $a_{t+1}$})
    &= 1 + \sum_{k=1}^{t} \Pr(S_k \text{ is occupied}) & \text{Linearity of expectation}\\
    &\leq 1 + \sum_{k=1}^{t} p_k(t) & \text{Definition of $p_k(t)$.}
\end{align*}
\end{proof}

\Cref{lem:insertioncost} reduces the problem to bounding $p_k(t)$. In \Cref{sect:recursive-bound}
we give a recursive expression for $F_S(t)$ and $p_k(t)$, 
then analyze $p_k(t)$ and the cost of insertion in \Cref{sect:expected-insertion-time}.

\section{The Recursive Bound}\label{sect:recursive-bound}

The $F_S(t)$ and $p_k(t)$ functions are difficult to evaluate directly, but easy to evaluate when the table is full.
Clearly $F_S(n)=n^{n}$ and $p_k(n)=1$.
\Cref{thm:recursivepcalculation}
allows us to work backward from these values,
by upper bounding $F_S(t-1)$ and $p_k(t-1)$
in terms of $F_S(t)$ and $p_k(t)$, respectively.

\begin{theorem}\label{thm:recursivepcalculation}

For any set $S$ with $|S|=k$, and any $k<t\leq n$,
\begin{align*}
nF_S(t-1)\leq F_S(t)\left(1-\dfrac{k}{t}+\dfrac{k}{n}\right).
\intertext{Moreover, we have}
p_k(t-1)\leq p_k(t)\left(1-\dfrac{k}{t}+\dfrac{k}{n}\right).
\end{align*}
\end{theorem}

Throughout this section $S$, $k=|S|$, and $t>k$ are fixed.
For any $A\in [n]^i$, $i\leq t$, we define the \emph{potential} $\phi(A)$ to be
\begin{align*}
\phi(A) &= n^{-(t-i)}\sum_{B\in [n]^{t-i}} \Occupy_S(AB).
%
\intertext{In other words, $\phi(A)$ is the expected number of occupied slots in $S$ when $A$ is extended by a random $(t-i)$-length suffix $B$.  Define $\delta$ to be the change in potential:}
\delta(A,e) &= \phi(Ae)-\phi(A).
\end{align*}

Remember that $\phi$ and $\delta$ are defined with respect to fixed $S,t$.

\begin{lemma}\label{lem:potentiallem}
The potential $\phi$ and change-of-potential $\delta$ 
satisfy the following.
\begin{enumerate}
\item For any $A\in [n]^t$, $\phi(A)=\Occupy_S(A)$.
\item Letting $\varepsilon$ be the empty sequence, $\phi(\varepsilon)=\frac{kt}{n}$.
\item For any $A\in [n]^{i}$, $i<t$, 
$\sum_{e\in [n]}\delta(A,e)=0$.
\item For every $e\in [n],$ $|\delta(A,e)|\leq 1$.
\end{enumerate}
\end{lemma}

\begin{proof}
\noindent\underline{Part 1.} Follows directly from the definition.\medskip

\noindent\underline{Part 2.} By symmetry the probability
that any particular slot is occupied after $t$ insertions is $t/n$.  The claim follows by linearity of expectation.\medskip

\noindent\underline{Part 3.} Observe that
$\phi(A) = n^{-1}\sum_e \phi(Ae)$,
so $\sum_e (\phi(Ae) - \phi(A)) = \sum_e \delta(A,e) = 0$.\medskip

\noindent\underline{Part 4.}
For any $B\in [n]^{t-i-1}$ and any $e'\in [n]$, 
\Cref{lem:Aoccupylem} implies
$|\Occupy_S(AeB)-\Occupy_S(Ae'B)|\leq 1$.
Thus,
\begin{align*}
|\delta(A,e)|
&\leq n^{-(t-i)}\sum_{e'\in [n],B\in [n]^{t-i-1}}|(\Occupy_S(AeB)-\Occupy_S(Ae'B))|\\
&\leq n^{-(t-i)}\sum_{e'\in [n],B\in [n]^{t-i-1}}1=1.
\end{align*}
\end{proof}

Recall that $F_S(t) = |\mathscr{H}_S(t)|$
is the number of length-$t$ sequences fully occupying $S$.  Define $\Delta$ as follows
\begin{align*}
\Delta &= \sum_{H\in\mathscr{H}_S(t)}(\phi(H)-\phi(\varepsilon))\\
    &= F_S(t)\left(k - \frac{kt}{n}\right) & \text{\Cref{lem:potentiallem}.}
\end{align*}

We could have calculated $\Delta$ differently, by taking each $H\in \mathscr{H}_S(t)$ and measuring
the change in potential as each element of 
$H=(h_1,h_2,\ldots,h_t)$ 
is appended to the sequence. Note that 
the sum of these values telescopes to 
$\phi(H) - \phi(\varepsilon)$.

\begin{align*}
\Delta &=\sum_{H\in\mathscr{H}_S(t)}\sum_{i=1}^t\delta(H^{<i},h_i)=\sum_{i=1}^t\sum_{H\in\mathscr{H}_S(t)}\delta(H^{<i},h_i).
\end{align*}

Now fix $i$ and consider the last sum $\sum_{H\in\mathscr{H}_S(t)} \delta(H^{<i},h_i)$, 
Recall that $\mathscr{H}_S^{-i}(t)$ is the subset of sequences in $\mathscr{H}_S(t)$ that still occupy $S$ when the $i$th element is removed.  From \Cref{lem:H-i} we know that $|\mathscr{H}_S^{-i}(t)| = nF_S(t-1)$.

\begin{lemma}\label{lem:Hsumdelta}
$\sum_{H\in \mathscr{H}_S^{-i}(t)}\delta(H^{<i},h_i)=0$.
\end{lemma}

\begin{proof}
Observe that if $H\in \mathscr{H}_S^{-i}(t)$
then every $H_0$ with $H^{-i} = H_0^{-i}$ is
also in $\mathscr{H}_S^{-i}(t)$.  
If we group the terms of the sum with a common
$H^{<i}$, each $h_i \in [n]$ appears an equal number of times.  By \Cref{lem:potentiallem}
the net contribution of each group is zero.
\end{proof}

Thus, to bound 
$\sum_{H\in\mathscr{H}_S(t)} \delta(H^{<i},h_i)$
it suffices to consider only sequences
in $\mathscr{H}_S(t)\setminus \mathscr{H}_S^{-i}(t)$.

\begin{lemma}\label{lem:deltabound}
$\sum_{H\in \mathscr{H}_S(t)}\delta(H^{<i},h_i)\leq F_S(t)-nF_S(t-1)$.
\end{lemma}
\begin{proof}
\begin{align*}
\sum_{H\in \mathscr{H}_S(t)}\delta(H^{<i},h_i)
&=\sum_{H\in \mathscr{H}_S(t)\setminus \mathscr{H}_S^{-i}(t)}\delta(H^{<i},h_i) 
& \text{\Cref{lem:Hsumdelta}}\\
&\leq \sum_{H\in \mathscr{H}_S(t)\setminus \mathscr{H}_S^{-i}(t)} 1
& \text{\Cref{lem:potentiallem}}\\
&= F_S(t)-nF_S(t-1) & \text{\Cref{lem:H-i}}
\end{align*}
\end{proof}

We now have all the pieces to assemble a proof of \Cref{thm:recursivepcalculation}.

\begin{proof}[Proof of \Cref{thm:recursivepcalculation}]
Using \Cref{lem:deltabound} to upper bound $\Delta$, 
we now have
\begin{align*}
\left(k-\frac{kt}{n}\right)F_S(t) &=\Delta \leq t(F_S(t)-nF_S(t-1))
\intertext{hence}
nF_S(t-1)&\leq \left(1-\frac{k}{t}+\frac{k}{n}\right)F_S(t).
\intertext{Since this bound applies to any $S$ with size $k$,
$p_k(t-1)$ is upper bounded in the same way.}
p_k(t-1) &\leq p_k(t)\left(1-\dfrac{k}{t}+\dfrac{k}{n}\right).
\end{align*}
\end{proof}

\section{Analysis of the Expected Insertion Time}\label{sect:expected-insertion-time}

By \Cref{lem:insertioncost} the expected insertion cost is bounded in terms of $p_k(t)$, and \Cref{thm:recursivepcalculation} gives us a recursive expression for $p_k(t)$.
In the base case we have
$p_k(n-1) = 1-k/n$.
After $n-1$ insertions there is one free slot, 
which by symmetry is outside any particular $k$-element set with probability $1-k/n$.  
Applying \Cref{thm:recursivepcalculation}, we bound $p_k(t)$ as follows,
where $d = n-t$.

\begin{align*}
p_k(t) 
&\leq \left(1-\frac{k}{n}\right)\prod_{j=t+1}^{n-1} \left(1-\frac{k}{j}+\frac{k}{n}\right) & \text{\Cref{thm:recursivepcalculation}}\\
&\leq \exp\left[-k\left(\frac{1}{n} + \sum_{j=t+1}^{n-1}\left(\frac{1}{j}-\frac{1}{n}\right)\right)\right]\\
&\leq \exp\left[-k\left(\frac{1}{n} + \sum_{\ell=1}^{d-1}\left(\frac{\ell}{n(n-\ell)}\right)\right)\right]\\
&\leq \exp\left[-k\left(\frac{1}{n} + \frac{d(d-1)}{2n^2}\right)\right]\\
&\leq \exp\left[-\frac{kd^2}{2n^2}\right].
\end{align*}
For $d=n-t = \epsilon n$, this is $\exp(-k\epsilon^2/2)$. Applying 
\Cref{lem:insertioncost}, 
the expected insertion cost is at most

\begin{align*}
1+\sum_{k=1}^n e^{-k\epsilon^2/2}
\leq 1+\int_0^{+\infty}e^{-x\epsilon^2/2} \,\mathrm{d} x
=1+\frac{2}{\epsilon^2}.
\end{align*}

\section{Conclusion}

An interesting open problem is to determine the 
independence necessary to implement fixed offset or
quadratic hashing that preserves the guarantees of this paper or~\cite{HuKLWYZ26b,KuszmaulX24,GuoPW27}.
It is known that 5-wise independence suffices for 
linear probing~\cite{PaghPR09}, but 4-wise 
independence is in general not enough~\cite{PatrascuT16}.

\bigskip

\paragraph{AI Acknowledgment.}
The original idea came from interacting with GPT-6 Astra.
The simplified presentation of the paper is due to 
the authors, who take full responsibility for the 
correctness of all statements and proofs.

\bibliographystyle{alpha}
\bibliography{references}

\end{document}